\documentclass[11pt]{llncs}
\usepackage{latexsym}
\usepackage{amsmath}
\usepackage{amssymb}
\usepackage[table]{xcolor}
\usepackage{graphicx}
\usepackage{verbatim}

\newcommand{\bbF}{\mathbb{F}}
\newcommand{\bbN}{\mathbb{N}}
\newcommand{\mctt}{MC_{TT}}
\newtheorem{assumption}{Assumption}
\title{On the Complexity of Computing Two Nonlinearity Measures}
\author{Magnus Gausdal Find}
\institute{Department of Mathematics and Computer Science\\
University of Southern Denmark}

\begin{document}
\maketitle

\begin{abstract}
We study the computational complexity of two Boolean nonlinearity measures:
the \emph{nonlinearity}
and the \emph{multiplicative complexity}.
We show that if one-way functions exist, no algorithm can compute the
multiplicative complexity in time $2^{O(n)}$ given the truth
table of length $2^n$, in fact under the same assumption
it is impossible to approximate the multiplicative complexity within
a factor of $(2-\epsilon)^{n/2}$.
When given a circuit, the problem of determining
the multiplicative complexity is in the second level of the polynomial
hierarchy.  For nonlinearity, we show that
it is $\#\mathbf{P}$ hard to compute given a function represented by a circuit.
\end{abstract}

\section{Introduction}
In many cryptographical settings, such as stream ciphers, block ciphers
and hashing, functions being used must be deterministic but
should somehow  ``look'' random.
Since these two desires are contradictory in nature,
one might settle with functions satisfying certain \emph{properties}
that
random Boolean functions possess with high probability.
One property is to be somehow  different from linear functions.
This can be quantitatively delineated using so called ``nonlinearity measures''.
 Two examples of nonlinearity measures are the
\emph{nonlinearity}, i.e. the Hamming distance
to the closest affine function, and
the \emph{multiplicative complexity}, i.e. the smallest number of AND
gates in a circuit over the basis $(\land,\oplus,1)$ computing the function.
For results relating these measures to each other
and cryptographic properties we refer to
\cite{carletbook,DBLP:conf/ciac/BoyarFP13}, and the references
therein.
The important point for this paper is that
there is a fair number of results on the form
``if $f$ has low value according to measure
$\mu$, $f$ is vulnerable to the following attack ...''.
Because of this, it was a design criteria in the 
Advanced Encryption Standard 
to have parts with high nonlinearity \cite{Standards2001}.
In a concrete situation, $f$ is an explicit, finite function, so
it is natural to ask how hard it is to compute $\mu$
given (some representation of) $f$. In this paper, the measure
$\mu$ will be either multiplicative complexity or nonlinearity.
We consider the two cases where $f$ is being represented by its
truth table, or by a circuit computing $f$.

We should emphasize that multiplicative complexity is an
interesting measure for other reasons than alone
being a measure of nonlinearity:
In many applications it is harder, in some sense,
to handle AND gates than XOR gates, so
one is interested in a circuit over $(\land,\oplus,1)$
with a small
number of AND gates, rather than a circuit with the
smallest number of \emph{gates}.
Examples of this include protocols for secure multiparty computation
 (see e.g. \cite{DBLP:conf/stoc/ChaumCD88,kolesnikov}),
non-interactive secure proofs of knowledge
\cite{DBLP:journals/joc/BoyarDP00},
and fully homomorphic encryption (see for example 
\cite{DBLP:conf/focs/Vaikuntanathan11}).

It is a main topic in several papers (see e.g. 
\cite{DBLP:journals/joc/BoyarMP13,DBLP:journals/jc/CenkO10,DBLP:journals/corr/abs-1108-2830}\footnote{Here we mean concrete finite functions, as opposed to
giving good (asymptotic) upper bounds for an infinite family of functions})
to find circuits with few AND gates for specific functions using
either exact or heuristic techniques.
Despite this and the applications mentioned above, it appears
that the  {computational} hardness has not been studied before.

The two measures have very different complexities, 
depending on the representation of $f$.

\paragraph{Organization of the Paper and Results}
In the following section, we introduce the problems and
necessary definitions. All our hardness results will
be based on assumptions stronger than $\mathbf{P}\neq \mathbf{NP}$,
more precisely the existence of pseudorandom function families
and the ``Strong Exponential Time Hypothesis''.
In Section~\ref{sec:ttinput} we show that
if pseudorandom function families exist, the multiplicative complexity
of a function represented by its truth table cannot be computed
(or even approximated with a factor $(2-\epsilon)^{n/2}$)
in polynomial time. This should be contrasted
to the well known fact that nonlinearity can be computed in almost linear
time using the Fast Walsh Transformation. In Section~\ref{sec:circuitinput},
we consider the problems when the function is represented
by a circuit.  We show that in terms of time complexity, under
our assumptions, the situations differ very little from
the case where the function is represented by a truth table.
However, in terms of complexity classes, the picture looks
quite different: Computing the nonlinearity is $\#\mathbf{P}$ hard,
and multiplicative complexity is in the second level of the polynomial
hierarchy.

\section{Preliminaries}
In the following, we let $\bbF_2$ be the finite field of size $2$ 
and $\bbF_2^n$ the $n$-dimensional vector space over $\bbF_2$.
We denote by $B_n$ the set of Boolean functions,
mapping from $\bbF_2^n$ into $\bbF_2$.
We say that $f\in B_n$ is
\emph{affine} if there exist $\mathbf{a}\in \bbF_2^n,c\in \bbF_2$
such that
$f(\mathbf{x})=\mathbf{a}\cdot \mathbf{x}+c$
and \emph{linear} if $f$ is affine with $f(\mathbf{0})=0$, with
arithmetic over $\bbF_2$. This gives the symbol ``$+$''
an overloaded meaning, since we also use it for
addition over the reals. It should be clear from the context, 
what is meant.

In the following an XOR-AND circuit is a circuit
with fanin $2$ over the basis $(\land,\oplus,1)$ (arithmetic
over $GF(2)$). All circuits from now on are assumed to be XOR-AND circuits.
We adopt standard terminology for circuits
(see e.g. \cite{wegener87}).
If nothing else is specified, for a circuit $C$ we let $n$
be the number of inputs and $m$ be the number of gates, 
which we refer to as the \emph{size} of $C$, denoted $|C|$.
For a circuit $C$ we let $f_C$ denote the function computed
by $C$, and $c_\wedge (C)$ denote the number of AND gates
in $C$.

For a function $f\in B_n$, the
\emph{multiplicative complexity} of $f$, denoted $c_\wedge(f)$,
is the smallest number of AND gates necessary and sufficient
in an XOR-AND circuit computing $f$. 
The \emph{nonlinearity} of a function $f$, denoted $NL(f)$ is the
Hamming distance to its closest affine function, more precisely
\[
NL(f) = 2^{n} -
\max_{\mathbf{a}\in \bbF_2^n,c\in\bbF_2}|
\{\mathbf{x}\in\bbF_2^{n}|f(\mathbf{x})=\mathbf{a}\cdot \mathbf{x}+ c \}
|.
\]

We consider four decision problems in this paper:
$NL_C$, $NL_{TT}$, $MC_C$ and $MC_{TT}$. For $NL_C$ (resp $MC_C$) the input
is a circuit and a target $s\in \mathbb{N}$ and the goal is to determine
whether the nonlinearity (resp. multiplicative complexity) of $f_C$ is at
most $s$.
For $NL_{TT}$ (resp. $\mctt$) the input is a truth
table of length $2^n$ of a function $f\in B_n$ and a target $s\in \bbN$, with
the goal to determine whether the nonlinearity (resp. multiplicative
complexity) of $f$ is at most $s$.

We let $a \in_R D$ denote that $a$ is distributed uniformly at random
from $D$.   We will need the following definition:
\begin{definition}
  A family of Boolean functions $f=\{f_{n}\}_{n\in \bbN}$,
$f_n\colon \{0,1 \}^n \times \{0,1\}^n \rightarrow \{0,1 \}$,
is a \emph{pseudorandom function family} if $f$ can be computed
in polynomial time
and for every probabilistic polynomial time oracle Turing
machine $A$,
\[
|\Pr_{k\in_R \{0,1\}^n}[A^{f_n(k,\cdot)}(1^n)=1] - 
\Pr_{g\in_R B_n}[A^{g(\cdot)}(1^n)=1]
 |
\leq
n^{-\omega(1)}.
\]
\end{definition}
Here $A^{H}$ denotes that the algorithm $A$ has oracle access
to a function $H$, that might be $f_n(\mathbf{k},\cdot)$ for
some $\mathbf{k}\in \bbF_2^n$ or a random $g\in B_n$, for
more details see \cite{DBLP:books/arorabarak}.
Some of our hardness results will be based on the following assumption.
\begin{assumption}
\label{ass:1way}
  There exist
pseudorandom function
families. 
\end{assumption}

It is known that pseudorandom function families exist
if one-way functions exist \cite{DBLP:journals/jacm/GoldreichGM86,DBLP:journals/siamcomp/HastadILL99,DBLP:books/arorabarak}, so we consider
Assumption~\ref{ass:1way} to be very plausible.
We will also use the following assumptions on the exponential
complexity of $SAT$, due to Impagliazzo and Paturi.
\begin{assumption}[Strong Exponential Time Hypothesis
\cite{DBLP:journals/jcss/ImpagliazzoP01}]
\label{ass:seth} 
For\\ any fixed $c<1$, no algorithm runs in time $2^{cn}$  
and computes $SAT$ correctly.
\end{assumption}

\section{Truth Table as Input}
\label{sec:ttinput}
It is a well known result that given a function $f\in B_n$ represented
by a truth table of length $2^n$, 
the nonlinearity can be computed using $O(n2^{n})$ basic arithmetic
operations. This is done using the ``Fast Walsh Transformation''
(See \cite{sloane1977theory} or chapter 1 in \cite{ryanodonnell}). 

  In this section we show that the situation is different for
multiplicative complexity:  Under Assumption~\ref{ass:1way},
$MC_{TT}$ cannot be computed in polynomial time.

In  \cite{KabanetsC00},  Kabanets and Cai showed
that if \emph{subexponentially} strong pseudorandom function
families exist, 
the Minimum Circuit Size Problem (MCSP) (the problem
of determining the size of a smallest circuit of a function given
its truth table) cannot be solved
in polynomial time. The proof goes by showing that if
MCSP could be solved in polynomial time this would
induce a natural combinatorial property (as
defined in \cite{dblp:journals/jcss/RazborovR97}) useful against
circuits of polynomial size. Now by the celebrated result of Razborov and Rudich 
\cite{dblp:journals/jcss/RazborovR97}, 
this implies the nonexistence of subexponential pseudorandom
function families.

Our proof below is similar
in that we use results from
\cite{boyar2000multiplicative} 
in a way similar to what is done in
 \cite{KabanetsC00,dblp:journals/jcss/RazborovR97} (see also
the excellent exposition in \cite{DBLP:books/arorabarak}). However
instead of showing the existence of a natural and useful
combinatorial property and appealing to limitations of natural proofs, we give
an explicit polynomial time algorithm for breaking any pseudorandom function
family, contradicting Assumption~\ref{ass:1way}.

\begin{theorem}
\label{hardnessofcomputeTT}
Under Assumption~\ref{ass:1way}, on input a truth
table of length $2^n$, $\mctt$  cannot be computed in
time  $2^{O(n)}$.
\end{theorem}

\begin{proof}
Let $\{ f_n\}_{n\in \bbN}$ be a pseudorandom function family.
Since $f$ is computable in polynomial time
it has circuits of polynomial size (see e.g. \cite{DBLP:books/arorabarak}),
so we can choose $c\geq 2$
such that $c_\wedge(f_n)\leq n^c$ for all $n\geq 2$.
Suppose for the sake of contradiction that
some algorithm computes $\mctt$ in time $2^{O(n)}$.
We now describe an algorithm that breaks the pseudorandom function family.
The algorithm has access to an oracle $H\in B_n$,
along with the promise either $H(\mathbf{x})=f_n(\mathbf{k},\mathbf{x})$ for
$\mathbf{k}\in_R \bbF_2^n$ or $H(\mathbf{x})=g(\mathbf{x})$
for $g\in_R B_n$. The goal of the algorithm is to 
distinguish between the two cases.
Specifically our algorithm will return 
$0$ if $H(\mathbf{x})=f(\mathbf{k},\mathbf{x})$
for some $\mathbf{k}\in \bbF_2^n$, and if $H(\mathbf{x})=g(\mathbf{x})$
it will return $1$ with high probability, where the probability
is only taken over the choice of $g$. 

Let $s=10c\log n$ and
define $h\in B_s$ as $h(\mathbf{x})=H(\mathbf{x}0^{n-s})$.
Obtain the complete truth table of $h$ by querying $H$ on all
the $2^{s}=2^{10c\log n}=n^{10c}$ points.
Now compute $c_\wedge(h)$. By assumption this can be done in time
$poly(n^{10c})$.
If $c_\wedge(h)>n^c$, output $1$,
otherwise output $0$. We now want to argue that this algorithm correctly
distinguishes between the two cases.
Suppose first that $H(\mathbf{x})=f_n(\mathbf{k},\cdot)$ for
some $\mathbf{k}\in \bbF_2^n$.
One can think of $h$ as $H$ where some of the input bits are fixed.
But in this case, $H$ can also be thought of as $f_n$ with
$n$ of the input bits fixed. Now take the circuit for $f_n$ with
the minimal number of AND gates.
Fixing the value of some of the input bits clearly
cannot increase the number of AND gates, hence
$c_\wedge(h) \leq c_\wedge(f_n)\leq n^c$.

Now it remains to argue that if $H$ is a random
function, we output $1$ with high probability.
We do this by using the following lemma.

\begin{lemma}[Boyar, Peralta, Pochuev]
\label{bfpcount}
  For all $s\geq 0$, the number of functions in $B_s$ that can
be computed with an XOR-AND circuit using at most $k$ AND gates
is at most $2^{k^2+2k+2ks+s+1}$.
\end{lemma}

If $g$ is random on $B_n$, then $h$ is random on $B_{10c\log n}$,
so the probability that $c_\wedge(h)\leq n^c$ is at most:
\[
\frac{2^{(n^c)^2+2(n^c)+2(n^c)(10c\log n)+10c\log n+1}}{  2^{2^{10c\log n}}}.
\]
This tends to $0$, so if $H$ is a random function the algorithm returns
$0$ with probability $o(1)$. In total we have
\[
|\Pr_{k\in_R \{0,1\}^n}[A^{f_n(k,\cdot)}(1^n)=1] - 
\Pr_{g\in_R B_n}[A^{g(\cdot)}(1^n)=1]
 |
=
|0 - (1- o(1))
 |,
\]
concluding that if the polynomial time algoritm for deciding
$\mctt$ exists, $f$ is not a pseudorandom function family.
\qed
\end{proof}

A common question to ask about a computationally
intractable problem is how well it can be \emph{approximated}
by a polynomial time algorithm.
An algorithm approximates $c_\wedge(f)$ with approximation
factor $\rho(n)$ if it always outputs
some value in the interval $[c_\wedge(f),\rho(n)c_\wedge(f)]$.
By refining the proof above, we see that it is hard to compute
$c_\wedge(f)$ within even a modest factor.

\begin{theorem}
\label{hardnessofapproxTT}
For every constant $\epsilon>0$, under Assumption~\ref{ass:1way},
no algorithm takes the $2^n$ bit truth table of a
function $f$ and approximates $c_\wedge(f)$
with  $\rho (n) \leq (2-\epsilon)^{n/2}$
in time $2^{O(n)}$.
\end{theorem}
\begin{proof}
Assume for the sake of contradiction that the algorithm $A$
violates the theorem.
The algorithm breaking any pseudorandom function family
works as the one in the previous proof, but instead
we return $1$ if the value returned by $A$
is at least $T=(n^c+1)\cdot(2-\epsilon)^{n/2}$.
Now arguments similar to those in the proof above show that
if $A$ returns a value larger than $T$, $H$ must be random,
and if $H$ is random, $h$ has multiplicative complexity at most
$(n^c+1)\cdot(2-\epsilon)^{n/2}$ with probability at most

\begin{align*}
\frac{
2^{\left( (n^c+1)\cdot(2-\epsilon)^{(10c\log n)/2}\right)^2        +
    2(n^c+1)\cdot(2-\epsilon)^{10c\log n/2} 10c\log n +10c\log n +1}
}{2^{2^{10c\log n}}}
\end{align*}
This tends to zero, implying that under the assumption on $A$, there is
no pseudorandom function family.
\qed 
\end{proof}

\section{Circuit as Input}
\label{sec:circuitinput}
From a practical point of view, the theorems \ref{hardnessofcomputeTT} and
\ref{hardnessofapproxTT} might seem unrealistic.
We are allowing the algorithm to be polynomial in the
length of the truth table, which is exponential in the number of variables.
However
most functions used for practical purposes admit small circuits.
To look at the entire truth table might (and in some cases
should) be infeasible.
When working with computational problems on circuits, it is somewhat
common to consider the running time in two parameters; the
number of inputs to the circuit, denoted by $n$, and the size of the circuit,
denoted by $m$. In the following we assume that $m$ is polynomial in $n$.
In this section we show that even
determining whether a circuit computes an affine function
is $\mathbf{coNP}$-complete.
In addition $NL_C$ can be computed in time $poly(m)2^n$, and
is $\# \mathbf{P}$-hard.
Under Assumption~\ref{ass:1way}, $MC_C$ cannot be computed in
time $poly(m)2^{O(n)}$, and is contained in the second level of
the polynomial hierarchy. In the following, we denote by $AFFINE$ the
set of circuits computing affine functions.

\begin{theorem}
\label{affineconp}
  $AFFINE$ is $\mathbf{coNP}$ complete.
\end{theorem}
\begin{proof}
First we show that it actually is in $\mathbf{coNP}$.
  Suppose $C\not\in AFFINE$.
Then if $f_C(\mathbf{0})=0$,
there exist $\mathbf{x},\mathbf{y}\in \bbF_2^n$ such that
$f_C(\mathbf{x}+\mathbf{y})\neq f_C(\mathbf{x})+f_C(\mathbf{y})$
and if $C(\mathbf{0})=1$, there exists
$\mathbf{x},\mathbf{y}$ such that
$C(\mathbf{x}+\mathbf{y})+1\neq C(\mathbf{x})+C(\mathbf{y})$.
Given $C,\mathbf{x}$ and $\mathbf{y}$ this can clearly be
computed in polynomial time.
To show hardness, we reduce from $TAUTOLOGY$, which is $\mathbf{coNP}$-complete.

Let $F$ be a formula on $n$ variables, $\mathbf{x}_1,\ldots,\mathbf{x}_n$.
Consider the following reduction:
First compute $c=F(0^n)$, then for every $\mathbf{e^{(i)}}$
(the vector with all coordinates 0 except the $i$th) compute
$F(\mathbf{e^{(i)}})$.
If any of these or $c$ are $0$, clearly $F\not\in TAUTOLOGY$, 
so we reduce to a circuit trivially not in $AFFINE$.
We claim that $F$ computes an affine function if and only if
$F\in TAUTOLOGY$. Suppose $F$ computes an affine function,
then
$F(\mathbf{x})=\mathbf{a}\cdot \mathbf{x}+c$ for some $\mathbf{a}\in \bbF_2^n$.
Then for every $\mathbf{e^{(i)}}$, we have
\[
F(\mathbf{e^{(i)}})=\mathbf{a}_i+1=1=F(\mathbf{0}),
\]
 so we must have that $\mathbf{a}=\mathbf{0}$, and $F$ is constant.
Conversely if it is not affine, it is certainly not
constant.
In particular it is not a tautology.\qed
\end{proof}  

So even determining whether the multiplicative complexity
or nonlinearity is $0$ is $\mathbf{coNP}$ complete.
In the light of the above reduction, any algorithm for $AFFINE$ induces an
algorithm for $SAT$ with essentially the same running time, so under
Assumption~\ref{ass:seth}, AFFINE needs time essentially $2^n$.
This should be contrasted with the fact that the seemingly harder
problem of computing $NL_C$ can be done in time $poly(m)2^n$
by first computing the entire truth table and then
using the Fast Walsh Transformation.
Despite the fact that $NL_C$ does not seem to require much more time
to compute than $AFFINE$, it is hard for a much larger
complexity class.

\begin{theorem}
  $NL_C$ is $\#\mathbf{P}$-hard.
\end{theorem}
\begin{proof}
  We reduce from $\# SAT$. Let the circuit $C$ on $n$
variables be an instance of $\# SAT$.
Consider the circuit $C'$ on $n+10$ variables, defined by
\[
C'(\mathbf{x}_1,\ldots , \mathbf{x}_{n+10})
=
C(\mathbf{x}_1,\ldots , \mathbf{x}_n)\wedge \mathbf{x}_{n+1} \wedge
\mathbf{x}_{n+2}\wedge \ldots \wedge
\mathbf{x}_{n+10}.
\]

First we claim that independently of $C$, the best affine approximation
of $f_{C'}$ is always $0$.
Notice that $0$ agrees with $f_{C'}$  whenever at 
least one of $x_{n+1},\ldots ,x_{n+10}$ is $0$, and when they are
all $1$ it agrees on $|\{\mathbf{x}\in \bbF_2^n | f_{C'}(\mathbf{x})=0\} |$
many points. In total $0$ and $f_{C'}$ agree on
\[
(2^{10}-1)2^{n}+|\{\mathbf{x}\in \bbF_2^n | f_{C}(\mathbf{x})=0\} |
\]
inputs. To see that any other affine function approximates $f_{C'}$
worse than $0$, 
notice that any nonconstant affine function is balanced
and thus has to disagree with $f_{C'}$ very often.
The nonlinearity of $f_{C'}$ is therefore
\begin{align*}
NL(f_{C'}) &= 2^{n+10} -
\max_{\mathbf{a}\in \bbF_2^{n+10},c\in\bbF_2}|
\{\mathbf{x}\in\bbF_2^{n+10}|f_{C'}(\mathbf{x})=\mathbf{a}\cdot \mathbf{x}+ c \}
|\\
&=
2^{n+10} -
|
\{\mathbf{x}\in\bbF_2^{n+10}|f_{C'}(\mathbf{x})=0 \}
|\\
&=
2^{n+10}
-
\left((2^{10}-1)2^{n}+|\{\mathbf{x}\in \bbF_2^n | f_{C}(\mathbf{x})=0\} |\right)\\
&=
2^{n}
-
|\{\mathbf{x}\in \bbF_2^n | f_{C}(\mathbf{x})=0\} |\\
&=
|\{\mathbf{x}\in \bbF_2^n | f_{C}(\mathbf{x})=1\} |\\
\end{align*}
So the nonlinearity of $f_{C'}$ equals
the number satisfying assignments for $C$.

\qed
\end{proof}

So letting the nonlinearity, $s$, be a part of the input for $NL_C$ changes
the problem from being in level $1$ of the polynomial hierarchy
to be $\# \mathbf{P}$ hard,
but does not seem to change the time complexity much. The situation
for $MC_C$ is essentially the opposite, under Assumption~\ref{ass:1way},
the time  $MC_C$ needs is strictly more time than $AFFINE$, but is 
contained in $\Sigma_2^p$. 
By appealing to Theorem~\ref{hardnessofcomputeTT} and
\ref{hardnessofapproxTT}, the following theorem follows.

\begin{theorem}
  Under Assumption~\ref{ass:1way}, no polynomial time algorithm computes
 $MC_C$. Furthermore no algorithm with
running time $poly(m)2^{O(n)}$ approximates $c_\wedge(f)$ with
a factor of $(2-\epsilon)^{n/2}$ for any constant $\epsilon > 0$.
\end{theorem}

We conclude by showing that although $MC_C$ under Assumption~\ref{ass:1way}
requires more time, it is nevertheless contained in
the second level of the polynomial hierarchy. 

\begin{theorem}
$MC_C \in \Sigma_2^p$.
\end{theorem}
\begin{proof}
First observe that $MC_C$ written as a language has the right form:
\[
MC_C = \{(C,s) | \exists C'\ \forall \mathbf{x}\in \bbF_2^n\ 
(C(\mathbf{x})=C'(\mathbf{x})\textrm{ and } c_\wedge(C') \leq s)\}.
\]
Now it only remains to show that one can choose the size of $C'$ is polynomial
in $n+|C|$. Specifically, for any $f\in B_n$, if $C'$ is the circuit
with the smallest number of AND gates computing $f$, for $n\geq 3$,
we can assume that
  $|C'|\leq 2(c_\wedge(f)+n)^2+c_\wedge$.
For notational convenience let $c_\wedge(f)=M$.
$C'$ consists of XOR and AND gates and
each of the $M$ AND gates has exactly two
inputs and one output. Consider some topological ordering of the AND gates,
and call the output of the $i$th AND gate $o_i$.
Each of the inputs to an AND gate is a sum (in $\bbF_2$) of
$x_i$s, $o_i$s and possibly the constant $1$.
Thus the $2M$ inputs to the AND gates and the output,
can be thought of as $2M+1$ sums over $\bbF_2$
over $n+M+1$ variables (we can think of the
constant $1$ as a variable with a hard-wired value). This can
be computed with at most
\[
 (2M+1)(n+M+1)
\leq  2(M+n)^2 
\]
XOR gates, where the inequality holds for $n\geq 3$.
Adding $c_\wedge(f)$ for the  AND gates, we get the claim.
The theorem now follows, since $c_\wedge(f)\leq |C|$\qed
\end{proof}

The relation between circuit size and multiplicative complexity
given in the proof above is not tight, and
we do not need it to be. See \cite{sergeevrelationship}
for a tight relationship.

\section*{Acknowledgements}
The author wishes to thank Joan Boyar for helpful discussions.

\bibliographystyle{splncs03}
\bibliography{refs}

 \end{document}